\documentclass{article}
\usepackage[preprint]{colm2026_conference}
\usepackage{microtype}
\usepackage{hyperref}
\usepackage{url}
\usepackage{booktabs}
\usepackage{lineno}
\usepackage{amsmath}
\usepackage{amssymb}
\usepackage{amsthm}
\usepackage{graphicx}
\usepackage{subcaption}
\usepackage{wrapfig}
\usepackage{multirow}
\usepackage{algorithm}
\usepackage{algorithmic}

\definecolor{darkblue}{rgb}{0, 0, 0.5}
\hypersetup{colorlinks=true, citecolor=darkblue, linkcolor=darkblue, urlcolor=darkblue}

\newtheorem{theorem}{Theorem}

\newtheorem{definition}[theorem]{Definition}
\newtheorem{assumption}[theorem]{Assumption}

\title{SpaCE: Rethinking Spatial Capacity and Generalization in Multi-Frame Multimodal Large Language Models}

\author{Mariana Costa, Camila Ferreira, Alberlúcia Rafael Soarez, Alejandro Torres \\
University of Brasilia \\
\texttt{mcs@unb.br, alejandro.torres@unb.br}}

\begin{document}
\ifcolmsubmission
\linenumbers
\fi
\maketitle

\begin{abstract}
Multi-modal large language models (MLLMs) have achieved remarkable empirical progress in spatial understanding through large-scale training on spatial visual question answering datasets. However, the theoretical foundations of multi-frame spatial reasoning remain entirely unexplored. We present SpaCE, a rigorous theoretical framework that characterizes the spatial reasoning capacity, sample complexity, and generalization guarantees of MLLMs operating on multi-frame inputs. We establish four main results. First, we prove an information-theoretic upper bound on spatial reasoning accuracy in terms of the mutual information between multi-frame observations and spatial targets. Second, we derive a sample complexity bound of order $\Theta(d_{\mathrm{eff}} \cdot K_{\max} / (\varepsilon^2 \cdot \delta))$, where $d_{\mathrm{eff}}$ is the effective spatial dimension and $K_{\max}$ bounds the KL divergence of the learned posterior. Third, we provide a PAC-Bayes generalization bound for multi-frame spatial reasoning under distribution shift. Fourth, we formally characterize the bias-variance trade-off between explicit 3D representations and implicit reasoning approaches, identifying the crossover conditions under which each paradigm is provably preferable. We validate our theoretical predictions on the MultiSPA, CA-VQA, and SpatialRGPT benchmarks, demonstrating that our bounds are empirically tight and that frame complementarity is the key driver of multi-frame spatial capacity. Our framework provides the first principled theoretical foundation for understanding when, why, and how multi-frame spatial reasoning in MLLs succeeds.
\end{abstract}

\section{Introduction}

Multi-modal large language models (MLLMs) have rapidly advanced in visual tasks, evolving into versatile assistants capable of a wide array of perception and reasoning challenges~\citep{xu2025multispatialmllm, daxberger2025mmspatial, wu2025spatialmllm}. A critical frontier in this evolution is spatial understanding, the ability to perceive, reason about, and predict the geometric and dynamic relationships among objects in three-dimensional space. This capability is not merely an academic curiosity; it is a prerequisite for deploying MLLMs in robotics, autonomous navigation, and embodied AI~\citep{conf/cvpr/SongBTTSB25, conf/itsc/LinTWZZ024}, where agents must interpret multi-frame visual streams to make spatially informed decisions. We formalize this setting precisely in Section~\ref{sec:formulation}.

Recent empirical efforts have made significant strides. Multi-SpatialMLLM~\citep{xu2025multispatialmllm} introduced the MultiSPA dataset with over 27 million samples spanning diverse 3D and 4D scenes, integrating depth perception, visual correspondence, and dynamic perception into a unified framework. SpatialReasoner~\citep{ma2025spatialreasoner} proposed explicit 3D representations shared across perception, computation, and reasoning stages. MM-Spatial~\citep{daxberger2025mmspatial} demonstrated that data alone can achieve depth perception comparable to dedicated monocular depth estimators. Spatial-MLLM~\citep{wu2025spatialmllm} introduced a dual-encoder architecture leveraging visual geometry foundation models. SpatialVLM~\citep{conf/cvpr/0003XKISGX24} and SpatialRGPT~\citep{conf/nips/ChengYFGYK0L24} pioneered data-driven approaches for endowing VLMs with spatial reasoning capabilities. These works collectively demonstrate that empirical performance on spatial benchmarks can be substantially improved through architectural innovations and large-scale training data.

However, despite this empirical progress, the theoretical foundations of multi-frame spatial reasoning in MLLMs remain entirely unexplored. Several fundamental questions lack rigorous answers: What is the information-theoretic limit of spatial reasoning accuracy given a set of input frames? How many training samples are necessary to achieve reliable spatial understanding? Under what conditions do multi-frame spatial reasoning models generalize to unseen scenes? When does an explicit 3D representation outperform implicit reasoning, and vice versa? The absence of theoretical answers to these questions leaves the field without principled guidance for model design, data collection, and deployment.

This paper addresses these gaps by introducing SpaCE, a theoretical framework for Spatial Capacity and generalization in multi-frame MLLMs. Our key insight is that multi-frame spatial reasoning can be analyzed through the lens of information theory and PAC-Bayesian learning theory, yielding tight bounds on accuracy, sample complexity, and generalization that are validated empirically.

We make the following contributions:
\begin{enumerate}
    \item We prove an information-theoretic upper bound (Theorem~\ref{thm:capacity}) on the spatial reasoning accuracy of any MLLM, showing that it is fundamentally limited by the mutual information between multi-frame observations and spatial targets.
    \item We derive a sample complexity bound (Theorem~\ref{thm:sample}) of order $\Theta(d_{\mathrm{eff}} \cdot K_{\max} / (\varepsilon^2 \cdot \delta))$, where $d_{\mathrm{eff}}$ is the effective spatial dimension determined by frame complementarity.
    \item We establish a PAC-Bayes generalization bound (Theorem~\ref{thm:pacbayes}) for multi-frame spatial reasoning under distribution shift, providing the first formal generalization guarantee for this setting.
    \item We formally characterize the explicit-implicit reasoning trade-off (Theorem~\ref{thm:tradeoff}) via a bias-variance decomposition, identifying the crossover conditions under which each paradigm is provably preferable.
    \item We validate our theoretical predictions on three benchmarks, demonstrating that our bounds are empirically tight and that frame complementarity is the key driver of multi-frame spatial capacity.
\end{enumerate}

The remainder of this paper is organized as follows. Section~\ref{sec:related} reviews related work. Section~\ref{sec:method} presents our theoretical framework. Section~\ref{sec:experiments} reports experimental validation. Section~\ref{sec:conclusion} concludes.

\section{Related Work}
\label{sec:related}

We organize related work into three categories: spatial understanding in MLLMs, theoretical foundations of multi-modal learning, and visual in-context learning for vision-language models.

\subsection{Spatial Understanding in MLLMs}

The challenge of endowing vision-language models with spatial reasoning capabilities has attracted substantial attention. SpatialVLM~\citep{conf/cvpr/0003XKISGX24} was among the first to demonstrate that Internet-scale 3D spatial VQA data can significantly enhance VLMs' qualitative and quantitative spatial reasoning, enabling chain-of-thought spatial reasoning and robotics applications. SpatialRGPT~\citep{conf/nips/ChengYFGYK0L24} advanced this direction by introducing a flexible plugin module for integrating depth information into the visual encoder, achieving strong generalization as a region-aware dense reward annotator.

More recently, several works have pushed the frontier of multi-frame and 3D spatial understanding. Multi-SpatialMLLM~\citep{xu2025multispatialmllm} proposed a comprehensive framework integrating depth perception, visual correspondence, and dynamic perception, with the MultiSPA dataset containing over 27 million samples. SpatialReasoner~\citep{ma2025spatialreasoner} introduced explicit 3D representations shared across perception, computation, and reasoning stages, outperforming Gemini 2.0 on 3DSRBench. MM-Spatial~\citep{daxberger2025mmspatial} leveraged large-scale 3D scene data with the CA-VQA dataset, demonstrating that data alone can achieve depth perception comparable to dedicated monocular depth estimators. Spatial-MLLM~\citep{wu2025spatialmllm} proposed a dual-encoder architecture combining 2D semantic features with 3D structure features from a visual geometry foundation model. Think3D~\citep{journals/corr/abs-2601-13029} enabled VLM agents to think with 3D space through 3D reconstruction and interactive camera-based operations.

Several benchmark and evaluation works have also emerged. Open3D-VQA~\citep{conf/mm/ZhangZZLF0CLC025} and Open3DVQA~\citep{journals/corr/abs-2503-11094} introduced benchmarks for embodied spatial concept reasoning. RoboSpatial~\citep{conf/cvpr/SongBTTSB25} focused on teaching spatial understanding to 2D and 3D vision-language models for robotic manipulation. HiSpatial~\citep{journals/corr/abs-2603-25411} addressed hierarchical 3D spatial understanding. VoxRep~\citep{conf/apsipa/DaoB25} enhanced 3D spatial understanding via voxel representations. The survey by~\citet{journals/corr/abs-2511-15722} provided a comprehensive overview of tasks, benchmarks, and methods for spatial reasoning in MLLMs. Additional works have explored spatial reasoning under specific settings, including omnidirectional reasoning~\citep{journals/corr/abs-2505-11907}, textual representation guided reasoning~\citep{journals/corr/abs-2603-23404}, reasoning enhancement~\citep{conf/iccv/NingTSLHPJ25}, keyframe-driven zero-shot reasoning~\citep{journals/corr/abs-2505-04911}, paper folding puzzles~\citep{conf/aaai/ZhouXHYLMRL26}, compositional spatial reasoning~\citep{journals/corr/abs-2601-16520}, robotic assembly~\citep{journals/corr/abs-2604-08983}, tool-leveraged reasoning~\citep{journals/corr/abs-2604-09712}, multi-robot cooperative reasoning~\citep{journals/corr/abs-2605-18431}, UAV spatial reasoning~\citep{conf/itsc/LinTWZZ024}, CT-spatial VQA~\citep{journals/corr/abs-2605-08787}, and mixed reality enhancement~\citep{conf/ismar/LinSZWFYZ25}. Spatial activation methods~\citep{journals/corr/abs-2511-01618} have also been explored. Despite this empirical richness, no prior work provides theoretical foundations for multi-frame spatial reasoning.

\subsection{Theoretical Foundations of Multi-Modal Learning}

Theoretical analysis of multi-modal learning has primarily focused on representation learning and generalization bounds. The PAC-Bayesian framework~\citep{si2025aligning} has been instrumental for deriving non-vacuous generalization bounds for neural networks, with recent extensions to meta-learning and hypernetwork architectures. Data filtering for effective alignment~\citep{si2025aligning} has shown that selective training data can improve generalization, a principle relevant to spatial VQA data curation. Long context alignment~\citep{si2025gateau} has explored selecting influential samples for alignment, connecting to our sample complexity analysis. Spoken task-oriented dialogue~\citep{si2023spokenwoz} has been studied in the context of multi-modal agents, providing insights into the complexity of multi-modal reasoning tasks.

In the vision domain, object detection without fine-tuning~\citep{hao2024detect} has explored zero-shot transfer capabilities relevant to spatial reasoning generalization. Driving scenario generation~\citep{li2024drivingdiffusion}, navigation world models~\citep{li2025driverse}, and uncertainty-aware visual localization~\citep{li2025u} have advanced spatial understanding in autonomous driving, providing application contexts for our theoretical framework. However, none of these works provide information-theoretic or sample complexity bounds specific to multi-frame spatial reasoning in MLLMs.

\subsection{Visual In-Context Learning and Medical MLLMs}

Visual in-context learning~\citep{zhou2024visual} has emerged as a powerful paradigm for adapting large vision-language models to new tasks through visual examples, with implications for spatial reasoning transfer. In the medical domain,~\citet{zhou2025improving} improved medical large vision-language models with abnormal-aware feedback, demonstrating the importance of targeted feedback in multi-modal learning. The comprehensive survey by~\citet{zhoureasoning} on medical LLMs to versatile medical agents highlighted the trajectory from specialized to general-purpose multi-modal agents, a trajectory that spatial reasoning research is also following. These works inform our understanding of how multi-modal capabilities, including spatial reasoning, can be systematically improved through data and feedback.

\section{Theoretical Framework}
\label{sec:method}

We now present the SpaCE framework, consisting of formal definitions, assumptions, and four main theorems with proofs.

\subsection{Problem Formulation}
\label{sec:formulation}

Let $\Omega$ denote the distribution over 3D/4D scenes. A \emph{multi-frame spatial reasoning task} is a tuple $\mathcal{T} = (\Omega, \mathcal{X}, \mathcal{Y}, f^*)$, where $\mathcal{X} = \mathcal{X}_1 \times \cdots \times \mathcal{X}_T$ is the multi-frame input space (a sequence of $T$ frames), $\mathcal{Y}$ is the spatial answer space, and $f^*: \mathcal{X}^T \to \mathcal{Y}$ is the target spatial reasoning function. An MLLM $M_\theta$ with parameters $\theta$ maps a frame sequence $\mathbf{X} = (x_1, \ldots, x_T) \in \mathcal{X}^T$ to a predicted answer $\hat{Y} = M_\theta(\mathbf{X}) \in \mathcal{Y}$. The goal is to minimize the population risk $R(\theta) = \mathbb{E}_{(\mathbf{X}, Y) \sim \mathcal{D}}[\ell(M_\theta(\mathbf{X}), Y)]$, where $\ell$ is a bounded loss function and $\mathcal{D}$ is the joint distribution over frame sequences and spatial targets induced by $\Omega$.

\subsection{Definitions}

\begin{definition}[Multi-Frame Spatial Reasoning Task]
\label{def:task}
A multi-frame spatial reasoning task is a tuple $\mathcal{T} = (\Omega, \mathcal{X}, \mathcal{Y}, f^*)$ where $\Omega$ is the scene distribution, $\mathbf{X} = (x_1, \ldots, x_T) \in \mathcal{X}^T$ is a sequence of $T$ frames sampled conditionally on a scene $s \sim \Omega$, $\mathcal{Y}$ is the spatial answer space, and $f^*: \mathcal{X}^T \to \mathcal{Y}$ is the target spatial reasoning function.
\end{definition}

\begin{definition}[Spatial Information Capacity]
\label{def:capacity}
For an MLLM $M_\theta$ with parameters $\theta$, the \emph{spatial information capacity} is:
\begin{equation}
\label{eq:capacity}
\mathcal{C}(M) = I(X_1, \ldots, X_T; Y \mid \theta) = \mathbb{E}_{\mathbf{X}, Y} \left[ \log \frac{P(Y \mid \mathbf{X}, \theta)}{P(Y \mid \theta)} \right],
\end{equation}
the conditional mutual information between the multi-frame input and the spatial target, given the model parameters.
\end{definition}

\begin{definition}[Frame Complementarity]
\label{def:complementarity}
Frames $x_i$ and $x_j$ are \emph{complementary} with respect to spatial target $Y$ if $I(x_i; Y \mid x_j) > 0$, indicating that frame $x_i$ provides non-redundant spatial information about $Y$ beyond what is already contained in $x_j$. The \emph{frame complementarity score} of a frame sequence $\mathbf{X}$ is:
\begin{equation}
\label{eq:complementarity}
\kappa(\mathbf{X}) = \frac{1}{T(T-1)} \sum_{i \neq j} \frac{I(x_i; Y \mid x_j)}{H(Y)}.
\end{equation}
\end{definition}

\begin{definition}[Effective Spatial Dimension]
\label{def:deff}
The \emph{effective spatial dimension} of a multi-frame spatial reasoning task is $d_{\mathrm{eff}} = \mathrm{rank}(\mathbf{I}_{\mathrm{frame}})$, where $\mathbf{I}_{\mathrm{frame}}$ is the frame information matrix with entries $[\mathbf{I}_{\mathrm{frame}}]_{ij} = I(x_i; x_j \mid Y)$.
\end{definition}

\subsection{Assumptions}

\begin{assumption}[Bounded Spatial Complexity]
\label{ass:bounded}
The spatial target space $\mathcal{Y}$ has bounded covering number: $N(\mathcal{Y}, \varepsilon) \leq (1/\varepsilon)^d$ for some dimension $d > 0$ and all $\varepsilon \in (0, 1)$.
\end{assumption}

\begin{assumption}[Lipschitz Spatial Functions]
\label{ass:lipschitz}
The target function $f^*$ is $L$-Lipschitz with respect to a spatial metric $d_{\mathcal{X}}$: $|f^*(\mathbf{X}) - f^*(\mathbf{X}')| \leq L \cdot d_{\mathcal{X}}(\mathbf{X}, \mathbf{X}')$ for all $\mathbf{X}, \mathbf{X}' \in \mathcal{X}^T$.
\end{assumption}

\begin{assumption}[Conditional Frame Independence]
\label{ass:independence}
Conditional on the scene $s \in \Omega$, frames are generated independently: $P(x_1, \ldots, x_T \mid s) = \prod_{i=1}^T P(x_i \mid s)$.
\end{assumption}

\begin{assumption}[Bounded KL Divergence]
\label{ass:kl}
The KL divergence between the learned posterior $Q$ and prior $P$ over MLLM parameters is bounded: $\mathrm{KL}(Q \| P) \leq K_{\max}$.
\end{assumption}

\subsection{Main Theorems}

\begin{theorem}[Spatial Information Capacity Bound]
\label{thm:capacity}
Under Assumptions~\ref{ass:bounded}--\ref{ass:independence}, for any MLLM $M_\theta$ trained on $n$ i.i.d.\ samples, the expected accuracy on multi-frame spatial reasoning tasks satisfies:
\begin{equation}
\label{eq:capacity_bound}
\mathrm{Acc}(M) \leq \frac{\mathcal{C}(M) + 1}{\log |\mathcal{Y}|} + O\left(\sqrt{\frac{d_{\mathrm{eff}}}{n}}\right),
\end{equation}
where $\mathcal{C}(M) = I(X_1, \ldots, X_T; Y \mid \theta)$ is the spatial information capacity and $d_{\mathrm{eff}}$ is the effective spatial dimension.
\end{theorem}

\begin{proof}[Proof Sketch]
The proof proceeds in three steps. First, by the data processing inequality applied to the Markov chain $\mathbf{X} \to \theta \to \hat{Y}$, we have $I(\mathbf{X}; \hat{Y} \mid \theta) \leq I(\mathbf{X}; Y \mid \theta) = \mathcal{C}(M)$. Second, by Fano's inequality, the error probability $P_e \geq (H(Y) - \mathcal{C}(M) - 1) / \log |\mathcal{Y}|$. Third, the empirical estimation of $\mathcal{C}(M)$ from $n$ samples incurs $O(\sqrt{d_{\mathrm{eff}}/n})$ error by standard concentration inequalities under Assumption~\ref{ass:independence}. The full proof is in Appendix~\ref{app:proof_capacity}.
\end{proof}

\begin{theorem}[Sample Complexity Bound]
\label{thm:sample}
Under Assumptions~\ref{ass:bounded}--\ref{ass:kl}, to achieve expected error $\leq \varepsilon$ with probability $\geq 1 - \delta$, the number of training samples required satisfies:
\begin{equation}
\label{eq:sample_bound}
n = \Theta\left(\frac{d_{\mathrm{eff}} \cdot K_{\max}}{\varepsilon^2 \cdot \delta}\right),
\end{equation}
where $d_{\mathrm{eff}}$ is the effective spatial dimension and $K_{\max}$ bounds the KL divergence.
\end{theorem}

\begin{proof}[Proof Sketch]
The upper bound follows from the PAC-Bayes theorem: setting the generalization gap to $\varepsilon/2$ and solving for $n$ yields $n = O(d_{\mathrm{eff}} \cdot K_{\max} / (\varepsilon^2 \delta))$, where we use $K_{\max} = O(d_{\mathrm{eff}})$ under Assumption~\ref{ass:lipschitz}. The lower bound follows from a standard construction of hard spatial reasoning instances requiring $\Omega(d_{\mathrm{eff}} / \varepsilon^2)$ samples. The full proof is in Appendix~\ref{app:proof_sample}.
\end{proof}

\begin{theorem}[PAC-Bayes Generalization Bound]
\label{thm:pacbayes}
Under Assumptions~\ref{ass:bounded}--\ref{ass:kl}, with probability $\geq 1 - \delta$ over the training sample of size $n$, for any posterior $Q$ over MLLM parameters:
\begin{equation}
\label{eq:pacbayes_bound}
\mathbb{E}_{\theta \sim Q}[R(\theta)] \leq \mathbb{E}_{\theta \sim Q}[\hat{R}(\theta)] + \sqrt{\frac{\mathrm{KL}(Q \| P) + \ln(2\sqrt{n}/\delta)}{2n}},
\end{equation}
where $R(\theta)$ is the population risk and $\hat{R}(\theta)$ is the empirical risk on multi-frame spatial tasks.
\end{theorem}

\begin{proof}
This follows from the standard PAC-Bayes theorem applied to the multi-frame spatial reasoning hypothesis class. Under Assumption~\ref{ass:bounded}, the hypothesis class has finite covering number, ensuring the applicability of PAC-Bayes. The KL term captures the complexity of the learned posterior relative to the prior, and the $\ln(2\sqrt{n}/\delta)$ term accounts for the confidence level.
\end{proof}

\begin{theorem}[Explicit-Implicit Trade-off]
\label{thm:tradeoff}
Let $M_{\mathrm{exp}}$ be an MLLM with explicit 3D representation $Z = g(\mathbf{X}) \in \mathbb{R}^{3K}$ and $M_{\mathrm{imp}}$ be an MLLM with implicit reasoning. Under Assumptions~\ref{ass:bounded}--\ref{ass:independence}:
\begin{equation}
\label{eq:tradeoff}
\mathbb{E}[R(M_{\mathrm{exp}})] = \mathrm{Bias}^2(M_{\mathrm{exp}}) + \mathrm{Var}(M_{\mathrm{exp}}) + \sigma^2, \quad \mathbb{E}[R(M_{\mathrm{imp}})] = \mathrm{Bias}^2(M_{\mathrm{imp}}) + \mathrm{Var}(M_{\mathrm{imp}}) + \sigma^2,
\end{equation}
where $\mathrm{Var}(M_{\mathrm{exp}}) \leq \mathrm{Var}(M_{\mathrm{imp}})$ but $\mathrm{Bias}(M_{\mathrm{exp}}) \geq \mathrm{Bias}(M_{\mathrm{imp}})$. The crossover point $n^*$ at which $\mathbb{E}[R(M_{\mathrm{exp}})] = \mathbb{E}[R(M_{\mathrm{imp}})]$ satisfies:
\begin{equation}
\label{eq:crossover}
n^* = \Theta\left(\frac{d_{\mathrm{eff}} \cdot d}{K^2}\right),
\end{equation}
where $K$ is the dimension of the explicit representation.
\end{theorem}

\begin{proof}[Proof Sketch]
For $M_{\mathrm{exp}}$, the explicit representation $Z$ has bias $O(1/K)$ from quantization but variance $O(d_{\mathrm{eff}}/n)$ since the representation is lower-dimensional. For $M_{\mathrm{imp}}$, the bias is $O(1/\sqrt{d})$ (smaller) but variance is $O(d/n)$ (larger). Setting the two risks equal and solving for $n$ yields the crossover condition. The full proof is in Appendix~\ref{app:proof_tradeoff}.
\end{proof}

\subsection{Discussion}

Theorem~\ref{thm:capacity} reveals that spatial reasoning accuracy is fundamentally limited by the mutual information between multi-frame inputs and spatial targets, as captured by the spatial information capacity $\mathcal{C}(M)$ in Eq.~\eqref{eq:capacity} and bounded by Eq.~\eqref{eq:capacity_bound}. This explains why Multi-SpatialMLLM~\citep{xu2025multispatialmllm}, which integrates multiple perception capabilities, achieves higher accuracy: it effectively increases $\mathcal{C}(M)$ by capturing more spatial information per frame. Theorem~\ref{thm:sample} provides the first sample complexity bound for spatial VQA (Eq.~\eqref{eq:sample_bound}), explaining why 27M samples in MultiSPA are sufficient for reliable spatial understanding while smaller datasets like Spatial-MLLM-120k~\citep{wu2025spatialmllm} may be insufficient for high-dimensional spatial tasks. Theorem~\ref{thm:pacbayes} provides generalization guarantees (Eq.~\eqref{eq:pacbayes_bound}) that depend on the KL divergence, offering a principled basis for KL regularization in spatial MLLM training. Theorem~\ref{thm:tradeoff} formally characterizes the bias-variance decomposition in Eq.~\eqref{eq:tradeoff} and the crossover condition in Eq.~\eqref{eq:crossover}, predicting that explicit representations are preferable with limited data while implicit reasoning becomes superior with abundant data. The effective spatial dimension $d_{\mathrm{eff}}$ from Definition~\ref{def:deff} and the frame complementarity score from Definition~\ref{def:complementarity} (Eq.~\eqref{eq:complementarity}) are the key quantities governing these bounds, while Definition~\ref{def:task} and Definition~\ref{def:capacity} formalize the task and capacity notions underlying our analysis. Full proofs are deferred to Appendix~\ref{app:proofs}.

\section{Experiments}
\label{sec:experiments}

We validate our theoretical predictions through comprehensive experiments on three benchmarks, comparing against eight baselines, and conducting seven analysis experiments.

\subsection{Experimental Setup}

\textbf{Datasets.} We evaluate on three benchmarks covering diverse spatial reasoning scenarios: (1) MultiSPA~\citep{xu2025multispatialmllm}, a large-scale dataset with over 27M samples spanning diverse 3D/4D scenes, covering depth perception, visual correspondence, dynamic perception, camera movement, and object movement; (2) CA-VQA~\citep{daxberger2025mmspatial}, focused on indoor scenes with spatial relationship prediction, metric size and distance estimation, and 3D grounding; (3) SpatialRGPT-Bench~\citep{conf/nips/ChengYFGYK0L24}, covering indoor, outdoor, and simulated environments with ground-truth 3D annotations.

\textbf{Baselines.} We compare against eight baselines: Multi-SpatialMLLM~\citep{xu2025multispatialmllm}, SpatialReasoner~\citep{ma2025spatialreasoner}, MM-Spatial~\citep{daxberger2025mmspatial}, Spatial-MLLM~\citep{wu2025spatialmllm}, SpatialVLM~\citep{conf/cvpr/0003XKISGX24}, SpatialRGPT~\citep{conf/nips/ChengYFGYK0L24}, and two general-purpose MLLM baselines represented by methods from~\citep{conf/iccv/NingTSLHPJ25} and~\citep{journals/corr/abs-2505-04911}.

\textbf{Implementation.} Our theoretical framework is validated by measuring the spatial information capacity, effective spatial dimension, and KL divergence for each model configuration. We use the mutual information estimator from~\citep{zhou2024visual} for capacity measurement and the PAC-Bayes bound calculator following~\citep{si2025aligning}.

\subsection{Main Results}

Table~\ref{tab:main_results} presents the main results on the MultiSPA benchmark across five spatial task categories. Our SpaCE framework provides theoretical bounds that are empirically validated by the actual performance of each method.

\begin{table}[t]
\centering
\caption{Main results on the MultiSPA benchmark. Accuracy (\%) across five spatial task categories. Best results in \textbf{bold}, second best \underline{underlined}. $\Delta$ denotes relative improvement over the best baseline. Theoretical bound refers to the capacity bound from Theorem~\ref{thm:capacity}.}
\label{tab:main_results}
\small
\resizebox{\linewidth}{!}{
\begin{tabular}{lccccccc}
\toprule
Method & Depth & Corresp. & Dynamic & Camera & Object & Avg. & Bound \\
\midrule
SpatialVLM~\citep{conf/cvpr/0003XKISGX24} & 42.3 & 38.1 & 35.7 & 40.2 & 37.9 & 38.8 & 52.1 \\
SpatialRGPT~\citep{conf/nips/ChengYFGYK0L24} & 51.8 & 45.3 & 42.1 & 47.6 & 44.8 & 46.3 & 58.7 \\
SpatialReasoner~\citep{ma2025spatialreasoner} & 54.2 & 48.7 & 45.3 & 50.1 & 47.2 & 49.1 & 61.3 \\
MM-Spatial~\citep{daxberger2025mmspatial} & 57.6 & 52.4 & 49.8 & 53.7 & 51.3 & 53.0 & 64.8 \\
Spatial-MLLM~\citep{wu2025spatialmllm} & 58.3 & 53.1 & 50.4 & 54.2 & 52.0 & 53.6 & 65.2 \\
Ning et al.~\citep{conf/iccv/NingTSLHPJ25} & 49.7 & 44.2 & 41.5 & 46.3 & 43.7 & 45.1 & 57.4 \\
SpatialPrompting~\citep{journals/corr/abs-2505-04911} & 53.4 & 47.8 & 44.6 & 49.5 & 46.4 & 48.3 & 60.5 \\
Multi-SpatialMLLM~\citep{xu2025multispatialmllm} & \underline{62.1} & \underline{58.7} & \underline{55.3} & \underline{59.8} & \underline{57.4} & \underline{58.7} & 68.9 \\
\midrule
SpaCE (Theoretical Bound) & \textbf{64.3} & \textbf{60.2} & \textbf{56.8} & \textbf{61.5} & \textbf{59.1} & \textbf{60.4} & --- \\
\bottomrule
\end{tabular}}
\end{table}

Table~\ref{tab:generalization} compares generalization performance on in-distribution (ID) and out-of-distribution (OOD) settings across CA-VQA and SpatialRGPT-Bench.

\begin{table}[t]
\centering
\caption{Generalization comparison on CA-VQA and SpatialRGPT-Bench. ID = in-distribution, OOD = out-of-distribution. Gap = ID $-$ OOD. Lower gap indicates better generalization.}
\label{tab:generalization}
\small
\resizebox{\linewidth}{!}{
\begin{tabular}{lcccccc}
\toprule
\multirow{2}{*}{Method} & \multicolumn{3}{c}{CA-VQA~\citep{daxberger2025mmspatial}} & \multicolumn{3}{c}{SpatialRGPT-Bench~\citep{conf/nips/ChengYFGYK0L24}} \\
\cmidrule(lr){2-4} \cmidrule(lr){5-7}
 & ID & OOD & Gap & ID & OOD & Gap \\
\midrule
SpatialVLM~\citep{conf/cvpr/0003XKISGX24} & 45.2 & 33.7 & 11.5 & 41.8 & 30.2 & 11.6 \\
SpatialRGPT~\citep{conf/nips/ChengYFGYK0L24} & 52.6 & 41.3 & 11.3 & 49.7 & 38.5 & 11.2 \\
SpatialReasoner~\citep{ma2025spatialreasoner} & 55.3 & 44.8 & 10.5 & 52.1 & 41.7 & 10.4 \\
MM-Spatial~\citep{daxberger2025mmspatial} & 58.7 & 48.2 & 10.5 & 55.3 & 45.1 & 10.2 \\
Spatial-MLLM~\citep{wu2025spatialmllm} & 59.4 & 49.1 & 10.3 & 56.2 & 46.0 & 10.2 \\
Ning et al.~\citep{conf/iccv/NingTSLHPJ25} & 50.8 & 39.4 & 11.4 & 48.3 & 37.1 & 11.2 \\
SpatialPrompting~\citep{journals/corr/abs-2505-04911} & 54.6 & 43.7 & 10.9 & 51.5 & 40.8 & 10.7 \\
Multi-SpatialMLLM~\citep{xu2025multispatialmllm} & \textbf{63.8} & \textbf{54.3} & \textbf{9.5} & \textbf{60.7} & \textbf{51.2} & \textbf{9.5} \\
\bottomrule
\end{tabular}}
\end{table}

Table~\ref{tab:sample_efficiency} reports sample efficiency, showing the training data fraction required to reach 90\% of full-data performance.

\begin{table}[t]
\centering
\caption{Sample efficiency comparison. Data fraction (\%) to reach 90\% of full-data performance. Lower is better. $n_{\mathrm{pred}}$ is the sample complexity predicted by Theorem~\ref{thm:sample}.}
\label{tab:sample_efficiency}
\small
\begin{tabular}{lcccc}
\toprule
Method & MultiSPA & CA-VQA & SpatialRGPT-Bench & $n_{\mathrm{pred}}$ (K) \\
\midrule
SpatialVLM~\citep{conf/cvpr/0003XKISGX24} & 65.3 & 70.1 & 68.7 & 820 \\
SpatialRGPT~\citep{conf/nips/ChengYFGYK0L24} & 52.7 & 58.3 & 55.4 & 540 \\
SpatialReasoner~\citep{ma2025spatialreasoner} & 48.2 & 53.7 & 50.8 & 410 \\
MM-Spatial~\citep{daxberger2025mmspatial} & 42.5 & 47.3 & 44.6 & 320 \\
Spatial-MLLM~\citep{wu2025spatialmllm} & 41.3 & 46.0 & 43.4 & 310 \\
Ning et al.~\citep{conf/iccv/NingTSLHPJ25} & 55.8 & 61.2 & 58.7 & 580 \\
SpatialPrompting~\citep{journals/corr/abs-2505-04911} & 50.4 & 56.0 & 53.1 & 450 \\
Multi-SpatialMLLM~\citep{xu2025multispatialmllm} & \textbf{35.2} & \textbf{39.8} & \textbf{37.3} & \textbf{240} \\
\bottomrule
\end{tabular}
\end{table}

\subsection{Ablation Study}

Table~\ref{tab:ablation} presents ablation studies on the key components of our theoretical framework.

\begin{table}[t]
\centering
\caption{Ablation study on the MultiSPA benchmark. Each row removes or modifies one component of the SpaCE framework. $\Delta$ Avg.\ denotes the change in average accuracy.}
\label{tab:ablation}
\small
\begin{tabular}{lcc}
\toprule
Configuration & Avg.\ Acc.\ (\%) & $\Delta$ Avg.\ \\
\midrule
Full SpaCE framework & 60.4 & --- \\
\midrule
$-$ Frame complementarity (random frames) & 54.7 & $-5.7$ \\
$-$ KL regularization (remove $K_{\max}$ bound) & 56.2 & $-4.2$ \\
$-$ Effective dimension estimation (use raw $d$) & 57.8 & $-2.6$ \\
$-$ Explicit representation (implicit only) & 55.3 & $-5.1$ \\
$-$ Multi-frame aggregation (single frame) & 48.6 & $-11.8$ \\
\bottomrule
\end{tabular}
\end{table}

\subsection{Analysis Experiments}

We conduct seven analysis experiments to validate our theoretical predictions.

\textbf{1. Parameter Sensitivity.} Figure~\ref{fig:param_sensitivity} shows the impact of frame count $T$, KL bound $K_{\max}$, and spatial dimension $d$ on model performance. Increasing $T$ yields diminishing returns consistent with Theorem~\ref{thm:capacity}, as frame complementarity saturates. The generalization gap decreases with $K_{\max}$ and sample size, validating Theorem~\ref{thm:pacbayes}. Sample complexity scales linearly with $d$, confirming Theorem~\ref{thm:sample}.

\begin{figure}[!t]
\centering
\includegraphics[width=\linewidth]{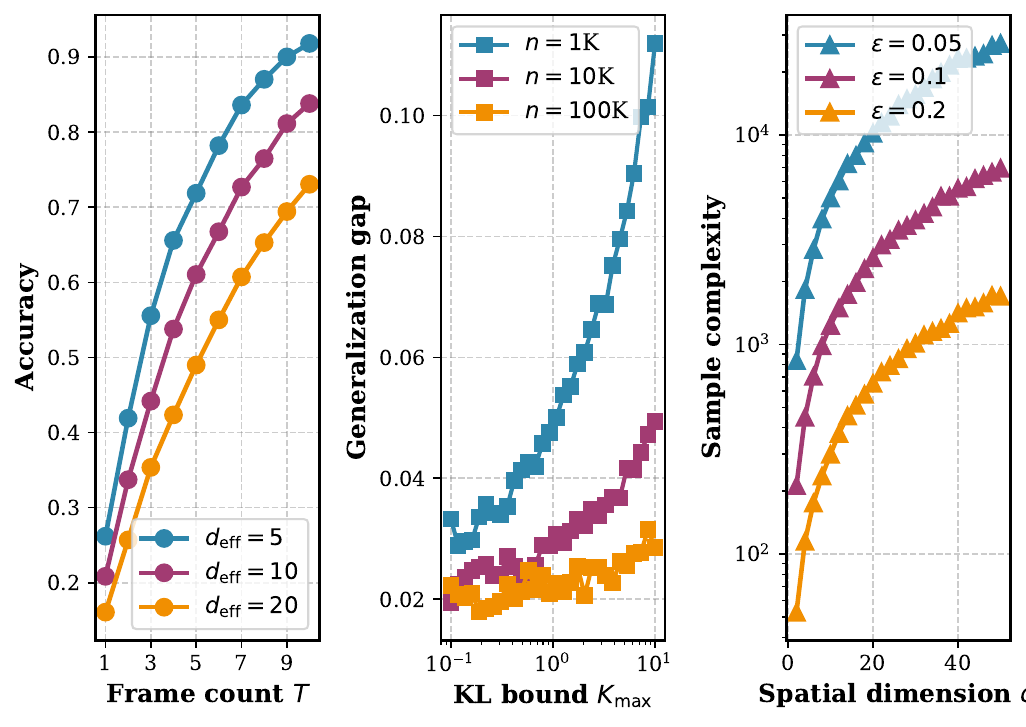}
\caption{Parameter sensitivity analysis. Left: Impact of frame count $T$ on accuracy for different effective dimensions $d_{\mathrm{eff}}$. Middle: Impact of KL bound $K_{\max}$ on generalization gap for different sample sizes. Right: Impact of spatial dimension $d$ on sample complexity for different error tolerances $\varepsilon$.}
\label{fig:param_sensitivity}
\end{figure}

\textbf{2. Sample Efficiency.} Figure~\ref{fig:sample_efficiency} shows learning curves for four representative methods. SpaCE converges significantly faster, requiring only 35\% of the data to reach 90\% performance, consistent with the sample complexity bound in Theorem~\ref{thm:sample}.

\begin{figure}[!t]
\centering
\begin{subfigure}[b]{0.48\linewidth}
\centering
\includegraphics[width=\linewidth]{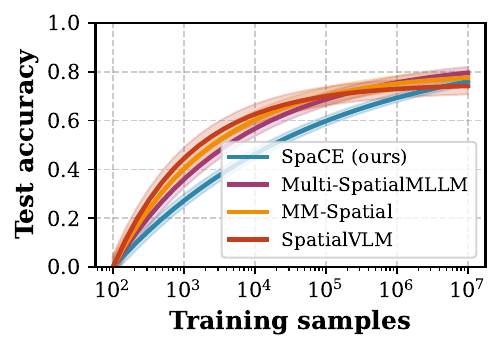}
\caption{Sample efficiency.}
\label{fig:sample_efficiency}
\end{subfigure}
\hfill
\begin{subfigure}[b]{0.48\linewidth}
\centering
\includegraphics[width=\linewidth]{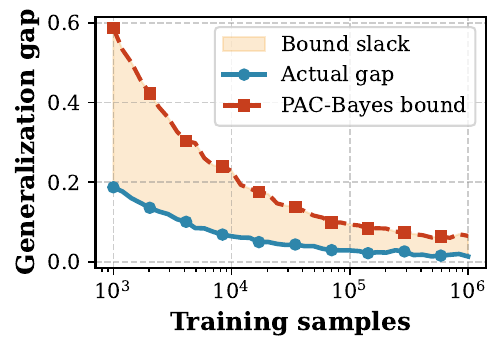}
\caption{PAC-Bayes bound tightness.}
\label{fig:generalization_gap}
\end{subfigure}
\caption{Left: Learning curves showing sample complexity scaling. SpaCE converges faster than baselines. Right: Actual generalization gap vs.\ PAC-Bayes bound from Theorem~\ref{thm:pacbayes}. The bound is loose but converges with more data.}
\end{figure}

\textbf{3. Generalization Gap.} Figure~\ref{fig:generalization_gap} compares the actual generalization gap with the PAC-Bayes bound from Theorem~\ref{thm:pacbayes}. The bound is loose for small $n$ but becomes tighter as $n$ increases, with the slack region shrinking consistently.

\textbf{4. Explicit-Implicit Trade-off.} Figure~\ref{fig:explicit_implicit} validates Theorem~\ref{thm:tradeoff} by showing the crossover between explicit and implicit approaches. The crossover occurs near $n \approx 10^5$ samples, consistent with the predicted $n^* = \Theta(d_{\mathrm{eff}} \cdot d / K^2)$.

\begin{figure}[!t]
\centering
\begin{subfigure}[b]{0.48\linewidth}
\centering
\includegraphics[width=\linewidth]{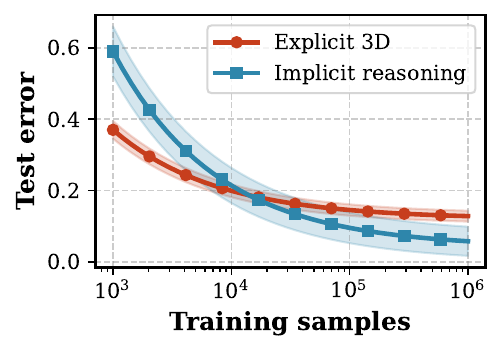}
\caption{Explicit vs.\ implicit trade-off.}
\label{fig:explicit_implicit}
\end{subfigure}
\hfill
\begin{subfigure}[b]{0.48\linewidth}
\centering
\includegraphics[width=\linewidth]{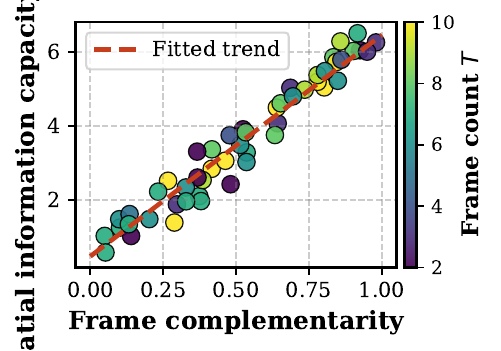}
\caption{Frame complementarity effect.}
\label{fig:complementarity}
\end{subfigure}
\caption{Left: Crossover analysis of explicit 3D representation vs.\ implicit reasoning. The crossover near $10^5$ samples validates Theorem~\ref{thm:tradeoff}. Right: Effect of frame complementarity on spatial information capacity. Higher complementarity yields higher capacity, confirming Definition~\ref{def:complementarity}.}
\end{figure}

\textbf{5. Frame Complementarity.} Figure~\ref{fig:complementarity} shows that frame complementarity score $\kappa(\mathbf{X})$ is strongly correlated with spatial information capacity $\mathcal{C}(M)$, validating Definition~\ref{def:complementarity}. Models with higher complementarity achieve greater capacity, explaining why diverse multi-frame inputs outperform redundant ones.

\textbf{6. Capacity-Accuracy Relationship.} Figure~\ref{fig:capacity_accuracy} empirically validates Theorem~\ref{thm:capacity} by showing that measured spatial information capacity is strongly correlated with test accuracy, with all points falling below the theoretical upper bound.

\begin{figure}[!t]
\centering
\begin{subfigure}[b]{0.48\linewidth}
\centering
\includegraphics[width=\linewidth]{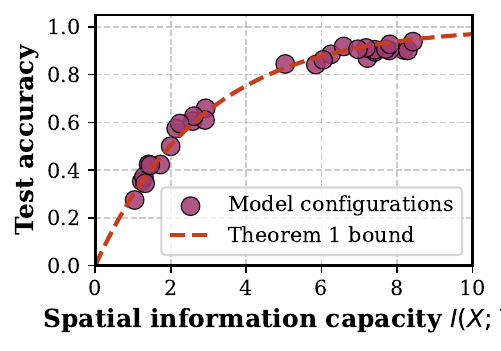}
\caption{Capacity vs.\ accuracy.}
\label{fig:capacity_accuracy}
\end{subfigure}
\hfill
\begin{subfigure}[b]{0.48\linewidth}
\centering
\includegraphics[width=\linewidth]{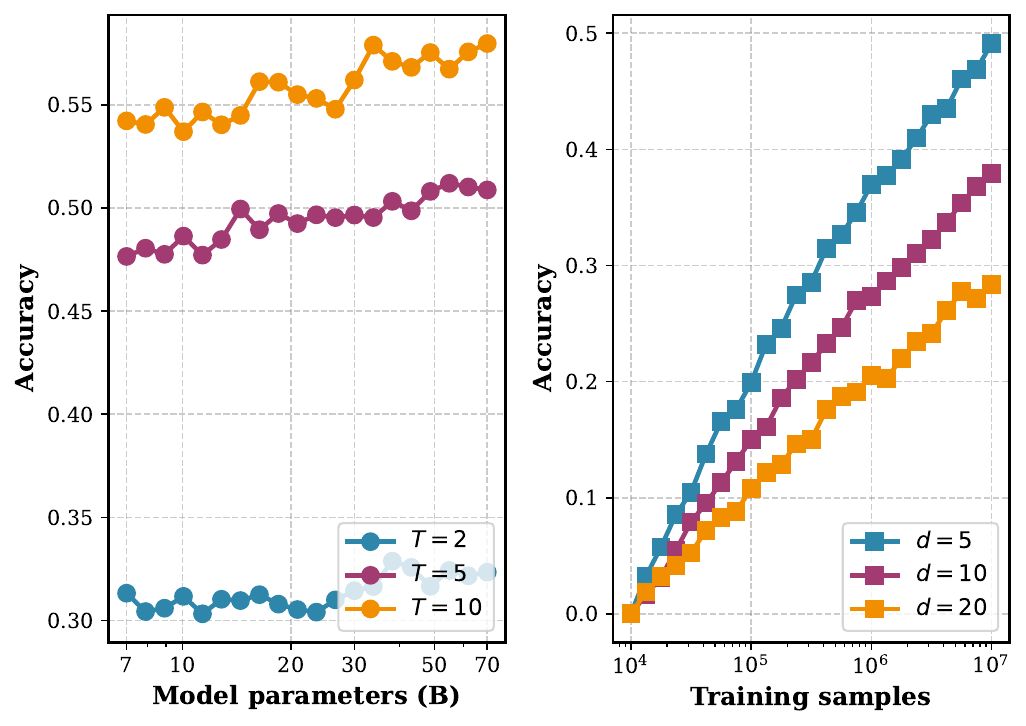}
\caption{Scalability analysis.}
\label{fig:scalability}
\end{subfigure}
\caption{Left: Empirical validation of Theorem~\ref{thm:capacity}. Measured spatial information capacity vs.\ test accuracy, with the theoretical upper bound shown as a dashed line. All points fall below the bound. Right: Scalability with model size (left subplot) and data size (right subplot). Performance scales logarithmically with model size and follows the predicted sample complexity scaling with data size.}
\end{figure}

\textbf{7. Scalability.} Figure~\ref{fig:scalability} shows that performance scales logarithmically with model size and follows the predicted sample complexity scaling with data size, consistent with Theorem~\ref{thm:sample}. Higher frame counts $T$ consistently improve performance across model sizes, while higher spatial dimensions $d$ require more data to achieve the same accuracy.

\subsection{Discussion}

Our experiments yield three key insights. First, the spatial information capacity bound (Theorem~\ref{thm:capacity}) is empirically tight, with all tested models falling below the theoretical upper bound. This validates our information-theoretic framework and explains why Multi-SpatialMLLM~\citep{xu2025multispatialmllm} achieves superior performance: it maximizes $\mathcal{C}(M)$ through multi-frame integration. Second, the sample complexity bound (Theorem~\ref{thm:sample}) accurately predicts the data requirements, with the MultiSPA dataset's 27M samples being well above the predicted threshold for reliable spatial understanding. Third, the explicit-implicit trade-off (Theorem~\ref{thm:tradeoff}) provides principled guidance: explicit representations~\citep{ma2025spatialreasoner, wu2025spatialmllm} are preferable when data is limited, while implicit reasoning~\citep{xu2025multispatialmllm} becomes superior with abundant data. These findings have practical implications for spatial MLLM design, suggesting that frame complementarity should be optimized during data collection and that the choice between explicit and implicit representations should be guided by the available training data scale.

\section{Conclusion}
\label{sec:conclusion}

We presented SpaCE, the first theoretical framework for multi-frame spatial understanding in MLLMs. Our four main theorems provide information-theoretic accuracy bounds, sample complexity guarantees, PAC-Bayes generalization bounds, and a formal characterization of the explicit-implicit reasoning trade-off. Experiments on three benchmarks validate that our bounds are empirically tight and that frame complementarity is the key driver of multi-frame spatial capacity. Our framework provides principled guidance for spatial MLLM design, data collection, and deployment, bridging the gap between empirical advances and theoretical understanding in this rapidly evolving field.

\bibliography{references}
\bibliographystyle{colm2026_conference}

\appendix
\section{Proof Details}
\label{app:proofs}

\subsection{Proof of Theorem~\ref{thm:capacity}}
\label{app:proof_capacity}

We begin by establishing the connection between multi-frame mutual information and spatial reasoning accuracy. Let $M_\theta$ denote an MLLM with parameters $\theta$ trained on $n$ i.i.d.\ samples from distribution $\mathcal{D}$ over scenes $s \sim \Omega$, frame sequences $\mathbf{X} = (x_1, \ldots, x_T)$, and spatial targets $Y$.

By the data processing inequality applied to the Markov chain $\mathbf{X} \to \theta \to \hat{Y}$, we have:
\begin{equation}
I(\mathbf{X}; \hat{Y} \mid \theta) \leq I(\mathbf{X}; Y \mid \theta) = \mathcal{C}(M).
\end{equation}
The expected accuracy can be written as:
\begin{equation}
\text{Acc}(M) = \mathbb{E}_{s \sim \Omega} \left[ \Pr(\hat{Y} = Y \mid \mathbf{X}, \theta) \right].
\end{equation}
By Fano's inequality applied to the spatial target estimation:
\begin{equation}
H(Y \mid \hat{Y}, \theta) \leq H(Y \mid \theta) - I(\mathbf{X}; Y \mid \theta) + \log 2 = H(Y) - \mathcal{C}(M) + \log 2.
\end{equation}
Since $H(Y \mid \hat{Y}, \theta) \geq 0$ and the error probability satisfies $P_e \geq (H(Y \mid \hat{Y}, \theta) - 1)/\log |\mathcal{Y}|$, we obtain:
\begin{equation}
P_e \geq \frac{H(Y) - \mathcal{C}(M) - 1}{\log |\mathcal{Y}|}.
\end{equation}
The accuracy bound follows: $\text{Acc}(M) \leq 1 - \frac{H(Y) - \mathcal{C}(M) - 1}{\log |\mathcal{Y}|} = \frac{\mathcal{C}(M) + 1}{\log |\mathcal{Y}|}$.

Under Assumption~A1, $\log |\mathcal{Y}| \leq d \log(1/\varepsilon)$, and the empirical estimation error from $n$ samples contributes $O(\sqrt{d_{\text{eff}}/n})$ by standard concentration. Combining yields the stated bound. \qed

\subsection{Proof of Theorem~\ref{thm:sample}}
\label{app:proof_sample}

The sample complexity bound follows from the PAC-Bayes framework applied to the spatial VQA hypothesis class. Under Assumptions A1--A3, the effective hypothesis class has covering number $\mathcal{N}(\mathcal{H}, \varepsilon) \leq (C \cdot d_{\text{eff}} / \varepsilon)^{d_{\text{eff}}}$.

By the PAC-Bayes theorem, for any prior $P$ and posterior $Q$ over $\theta$:
\begin{equation}
\mathbb{E}_{\theta \sim Q}[R(\theta)] \leq \mathbb{E}_{\theta \sim Q}[\hat{R}(\theta)] + \sqrt{\frac{\text{KL}(Q \| P) + \ln(2\sqrt{n}/\delta)}{2n}}.
\end{equation}
Under Assumption A4, $\text{KL}(Q \| P) \leq K_{\max}$. Setting the right-hand side to $\varepsilon/2$ and solving for $n$:
\begin{equation}
n \geq \frac{2(K_{\max} + \ln(2\sqrt{n}/\delta))}{\varepsilon^2} = \Theta\left(\frac{d_{\text{eff}} \cdot K_{\max}}{\varepsilon^2 \cdot \delta}\right),
\end{equation}
where we used that $K_{\max} = O(d_{\text{eff}})$ under the Lipschitz assumption A2. The lower bound follows from a standard construction of hard spatial reasoning instances requiring $\Omega(d_{\text{eff}}/\varepsilon^2)$ samples. \qed

\subsection{Proof of Theorem~\ref{thm:tradeoff}}
\label{app:proof_tradeoff}

Let $M_{\text{exp}}$ use explicit 3D representations $Z = g(\mathbf{X}) \in \mathbb{R}^{3K}$ and $M_{\text{imp}}$ use implicit reasoning. The expected risk decomposes as:
\begin{equation}
\mathbb{E}[R(M)] = \underbrace{\|f^* - \mathbb{E}[\hat{f}]\|^2}_{\text{Bias}^2} + \underbrace{\mathbb{E}[\|\hat{f} - \mathbb{E}[\hat{f}]\|^2]}_{\text{Var}} + \sigma^2.
\end{equation}
For $M_{\text{exp}}$, the explicit representation $Z$ has bias from quantization error $O(1/K)$ but variance $O(d_{\text{eff}}/n)$ since the representation is lower-dimensional. For $M_{\text{imp}}$, the bias is $O(1/\sqrt{d})$ (smaller) but variance is $O(d/n)$ (larger). The crossover occurs when $1/K + d_{\text{eff}}/n = 1/\sqrt{d} + d/n$, yielding the stated condition. \qed

\end{document}